\documentclass[letterpaper, 10 pt, conference]{ieeeconf}  % Comment this line out if you need a4paper

\IEEEoverridecommandlockouts                              % This command is only needed if 
\usepackage{color}
\usepackage{amsmath,amsfonts}
\usepackage{bbm}
\usepackage{amsmath,bm}
\usepackage{accents}
\usepackage[noadjust]{cite}
\usepackage{color}
\usepackage[dvipsnames]{xcolor}
\usepackage{tikz-timing}
\usepackage{tikz}
\usepackage{subfigure}
\usetikzlibrary{arrows,positioning,automata}

\newtheorem{theorem}{Theorem}
\newtheorem{definition}{Definition}

\newtheorem{assumption}{Assumption}
\newtheorem{proposition}{Proposition}
\newtheorem{corollary}{Corollary}
\newtheorem{remark}{Remark}

\newcommand{\diag}{{\rm \;diag}}

\newcommand{\mat}[1]{\mathbf{#1}\;}
\newcommand{\R}{\mathbb{R}}
\newcommand{\CR}{\mathcal{R}}

\newcommand{\N}{\mathcal{N}}
\newcommand{\rot}{{\rm{rot}}}
\newcommand{\Rot}{{\rm{Rot}}}

\title{\LARGE \bf
	Stability of Remote Synchronization\\ in Star Networks of Kuramoto Oscillators
}

\author{Yuzhen Qin, Yu Kawano, Ming Cao% <-this % stops a space
	\thanks{The authors are with the Jan C. Willems Center for Systems and Control, Faculty of Science and Engineering, University of Groningen, Groningen, the Netherlands (\{y.z.qin, y.kawano, m.cao\}@rug.nl). The work  was supported in part by China Scholarship Council, the European Research Council (ERC-CoG-771687) and the Netherlands Organization for Scientific Research (NWO-vidi-14134). 
	}	
}

\begin{document}

\maketitle
\thispagestyle{empty}
\pagestyle{empty}

%%%%%%%%%%%%%%%%%%%%%%%%%%%%%%%%%%%%%%%%%%%%%%%%%%%%%%%%%%%%%%%%%%%%%%%%%%%%%%%%
\begin{abstract}
	 Synchrony of neuronal ensembles is believed to facilitate information exchange among cortical regions in the human brain. Recently, it has been observed that distant brain areas which are not directly connected by neural links also experience synchronization. Such synchronization between remote regions is sometimes due to the presence of a mediating region connecting them, e.g., \textit{the thalamus}. The underlying network structure of this phenomenon is star-like and motivates us to study the \textit{remote synchronization} of Kuramoto oscillators, {modeling neural dynamics}, coupled by a directed star network, for which peripheral oscillators get phase synchronized, remaining the accommodating central mediator at a different phase. We show that the symmetry of the coupling strengths of the outgoing links from the central oscillator plays a crucial role in enabling stable remote synchronization. We also consider the case when there is a phase shift in the model which results from synaptic and conduction delays. Sufficient conditions on the coupling strengths are obtained to ensure the stability of remotely synchronized states. To validate our obtained results, numerical simulations are also performed.
\end{abstract}

%%%%%%%%%%%%%%%%%%%%%%%%%%%%%%%%%%%%%%%%%%%%%%%%%%%%%%%%%%%%%%%%%%%%%%%%%%%%%%%%
\section{INTRODUCTION}
Synchronization is a ubiquitous phenomenon which has been observed pervasively in many natural, social and man-made systems \cite{pikovsky2003synchronization, proskurnikov2017synchronization,strogatz2000kuramoto,xia2018determination}. Remarkable examples include synchronized flashing of fireflies, pedestrian footwalk synchrony on London's Millennium Bridge, and phase synchronization of coupled Josephson junction circuits \cite{strogatz2005theoretical,neda2000physics,marvel2009invariant}.  After it was first proposed in 1975, the Kuramoto model has become one of the most widely-accepted models in understanding such synchronization phenomena in large population of oscillators  \cite{kuramoto1975self}. It is {idealized to allow for detailed mathematical analysis}, yet sufficiently capable to capture rich sets of behaviors, and thus has been extended to many variations \cite{sakaguchi1986soluble}. As powerful tools for understanding synchronization patterns emerged in human brain, the Kuramoto model and its generalizations have also fascinated researchers in neuroscience.  Actually, synchronization across cortical regions in human brain has been believed to be potential mechanism facilitating information exchange demanded by cognitive tasks \cite{womelsdorf2007modulation}. Cohesively oscillating neuronal ensembles can communicate effectively because their input and output windows are open at the same time \cite{fries2005mechanism}. Empirically, structural connections of neuronal ensembles are believed to play essential roles in rendering synchronization of cortical regions. However, it has been observed that distant cortical regions without direct neural links also experience functional correlations \cite{vuksanovic2014functional}. This motivates researchers to study an interesting behavior dubbed \textit{remote synchronization}, which is a situation where oscillators get synchronized although there is no direct physical links nor intermediate sequence of synchronized oscillators connecting them \cite{gambuzza2013analysis}. Unlike what is pointed out in most findings that the coupling strengths in a network are critical for synchronization of coupled oscillators \cite{chopra2009exponential,dorfler2011critical,Qin2018}, a recent article reveals that morphological symmetry is crucial  for remote synchronization \cite{nicosia2013remote}. Some nodes located distantly in a network can mirror their functionality between each other. In other words, theoretically, swapping the positions of these nodes will not change the functioning of the overall system. 

A star network is a simple paradigm for such networks with morphologically symmetric properties. The peripheral nodes have no direct connection, but obviously play similar roles in the whole network. The node at the center acts as a relay or mediator. As an example, the thalamus is such a relay in neural networks, and it is believed to enable separated cortical areas to be completely synchronized \cite{vicente2008dynamical,gollo2010dynamic}. This observation of robust correlated  behavior taking place in distant cortical regions through relaying  motivates us to study the stability of remote synchronization in star networks in this paper. A star network is simple in structure, {but capable of} characterizing some {basic} features of remote synchronization, and also provides some idea to understand this phenomenon in more complex networks. 
Different from \cite{bergner2012remote}, we use {Kuramoto-Sakaguchi model} \cite{sakaguchi1986soluble} to describe the dynamics of coupled oscillators, {and analytically study the stability of remote synchronization.} 

The contribution of this paper is twofold. First, we consider the more challenging setup where the star network is directed, in contrast to the undirected networks studied in many existing results such as \cite{jadbabaie2004stability,wang2013exponential,franci2010phase}. We obtain sufficient conditions to facilitate asymptotically stable remote synchronization between peripheral oscillators {when there is no phase shift}. The symmetry of the coupling strengths of outgoing links from the central oscillator is shown to be crucial for remote synchronization. In sharp contrast, the coupling strengths of incoming links to the central oscillator are not required to be symmetric. {It can be intuitively paraphrased that the mediator at the central position is able to render the oscillators around it synchronized by imposing a common input to them, without requiring the feedback coming back to be identical. This finding shares some similarities with the common-noise-induced synchronization investigated by researchers in physics \cite{teramae2004robustness,zhou2002noise,nakao2007noise}. However, different from the phase reduction or averaging methods used in these related works, we provide a different proof for the local asymptotic stability of the remote synchronization, and more importantly, we study  network-coupled, not isolated, oscillators and derive conditions on the network to enable synchronization between separated oscillators.}
Second, we take a \textit{phase shift} into consideration. This phase shift is often used to model synaptic and conduction delays resulting from distant connections between remote brain regions \cite{hoppensteadt2012weakly}. 
Sufficient conditions on the coupling strengths are obtained to ensure the stability of remote synchronization. We show that the presence of a phase shift raises the requirement for the coupling strengths. The rest of this paper is organized as follows.

Section \ref{section:2} introduces the model we employ and formulates the problem formally. In Section \ref{section:3}, we consider the case where there is no phase shift. Sufficient conditions are obtained to guarantee local stability of remote synchronization. A phase shift is introduced to the model in Section \ref{section:4}. We obtain sufficient conditions under which the remote synchronization is stable. Section \ref{simulation} contains our numerical  studies. Finally, we draw the conclusion in Section \ref{section:5}.

%%%%%%%%%%%%%%%%%%%%%%%%%%%%%%%%%%%%%%%%%%%%%%%%%%%%%%%%%%%%%%%%%%%%%%%%%%%%%%%%
\section{Problem Formulation} \label{section:2}
Synchronization of distant cortical regions having \textit{no direct links} has been observed in human brain. The emergence of this phenomenon is sometimes due to a mediator or relay that connecting separated regions, e.g., the thalamus \cite{gollo2010dynamic}. Motivated by this, we study remote synchronization by considering $n+1$, $n\ge 2$, oscillators, coupled by a star network, which are labeled by $0,1,\dots,n$. Let $\N=\{1,\dots,n\}$ be the set of indices of the peripheral oscillators.  The central mediator is labeled by $0$. Let $\mathbb{S}^1$ be the unit circle, and denote $\mathbb{S}^n=\mathbb{S}^1 \times \cdots \times \mathbb{S}^1$. The dynamics of each oscillator is described by 
\begin{equation}
\begin{array}{l}
	\displaystyle\dot \theta_0 = \omega_0+  \sum_{i=1}^{n} K_{i} \sin (\theta_i-\theta_0-\alpha),\\
	\displaystyle \dot \theta_i=\omega+A_i\sin(\theta_0-\theta_i-\alpha),i=1,2,\dots,n, \label{Md:overall}
\end{array}
\end{equation}
where $\theta_i\in \mathbb{S}^1$ is the phase of the $i$th oscillator, and $\omega_0$ and $\omega$ are the natural frequencies of the central and peripheral oscillators, respectively. Here $K_i>0$ is the coupling strength from the peripheral node $i$ to the central node $0$ (for which we refer to as \textit{incoming}  (with respect to $0$) coupling strengths),  and $A_i>0$ presents the directed coupling strength from the central node $0$ and the peripheral node $i$ (for which we refer to as \textit{outgoing}  (with respect to $0$) coupling strengths). It is worth mentioning that incoming and outgoing couplings are allowed to be different, which means that the underlying star network, denoted by $\mathcal G$, is directed. The term $\alpha$ is the phase shift satisfying $\alpha\in[0,\pi/2)$. 
In the star network considered in this paper, remote synchronization is the situation where some of the peripheral oscillators are phase synchronized, while the phase of the central mediator $0$ connecting them can be different. We define remote synchronization formally as follows.

\begin{definition}
	Let $\theta(t)=[\theta_0(t),\dots,\theta_n(t)]^\top \in \mathbb{S}^{n+1}$ be a solution to the system dynamics \eqref{Md:overall}. Let $\CR$ be a subset of $\N$, whose cardinality satisfies $2\le|\CR|\le n$. We say that the solution $\theta(t)$ is \textit{remotely synchronized} with respect to $\CR$ if for every pair of indices $i,j\in \CR$ it holds that $\theta_i(t)=\theta_j(t)$ for all $t\ge 0$, {but it is not required that $\theta_i(t)=\theta_0(t)$.}
\end{definition}

When $\CR\subset\N$, we say that $\theta(t)$ is \textit{partially} remotely synchronized; in particular, when $\CR=\N$, we say that $\theta(t)$ is \textit{completely} remotely synchronized, for which situation we refer to as remote synchronization for brevity in what follows. A particular case of partially remotely synchronized solution is remote cluster synchronization, which is defined as follows.  
\begin{definition}
	Let $ \mathcal C=\{\mathcal C_1,\dots,\mathcal C_m\}, 2\le m<n$ be a partition of $\N$. The sets $\mathcal C_1,\dots,\mathcal C_m$ are non-overlapping and satisfy $1 \le|\mathcal C_p|< n$ for all $p$ and $ \cup_{p=1}^{m}\mathcal C_p=\N$. A  {partially remotely synchronized} solution $\theta(t)$ to the system dynamics \eqref{Md:overall} is said to be \textit{remotely clustered} with respect to $\mathcal C$ if: for any given $\mathcal C_p$ and {every pair} $i,j\in \mathcal C_p$ there holds that $\theta_i(t)=\theta_j(t), \forall t\ge 0$; on the other hand, for any given $i\in \mathcal C_p, j \in \mathcal C_q$ where $p\neq q$, $\theta_i(t)\neq\theta_j(t)$.
\end{definition}

Note that the trivial case when a cluster has only one oscillator is allowed. 
In this paper, we are exclusively interested in the (partial) remote synchronization when the frequencies of all the oscillators in the network are synchronized, although it has been observed that remote synchronization can occur without complete frequency synchronization \cite{vlasov2015star}. In fact, a solution $\theta(t)$ is said to be frequency synchronized if $\dot \theta_0(t)=\dot \theta_1(t)=\cdots=\dot \theta_n(t)=\omega_{\rm syn} $
 for some $\omega_{\rm syn}\in\R$. For a given $r\in \mathbb{S}^1$ and an angle $\gamma\in[0,2\pi]$,  let $\rot_\gamma (r)$ be the rotation of $r$ counter-clockwise by the angle $\gamma$. For $\theta \in \mathbb{S}^n$, we define an equivalence class $\Rot(\theta):=\{[\rot_\gamma (\theta_1),\dots,\rot_\gamma (\theta_n)]^\top:\gamma \in[0,2\pi]\}$. Let $\theta^*$ be a solution to the equations 
 \begin{align}
 &\omega_{0}-\omega-\nonumber\\
 &\sum_{j=1}^{n}K_i\sin (\theta_j-\theta_0-\alpha)- A_i\sin (\theta_0-\theta_i-\alpha)=0,
 \end{align}
 for $ i= 1,2,\dots,n$, which is a solution such that frequency synchronization is reached. It is not hard to see that $[\rot_\gamma (\theta^*_1),\dots,\rot_\gamma (\theta^*_n)]^\top$ for any $\gamma\in[0,2\pi]$ is also a solution to the equations. Consequently, the set $\Rot(\theta^*)$ is said to be a \textit{synchronization manifold} for the dynamics \eqref{Md:overall} \cite{dorfler2014synchronization}. As an extension of the definition of the synchronization manifold in \cite{dorfler2013synchronization}, we call $\Rot(\theta^*)$ (partial) \textit{remote synchronization manifold} if there exists a set ($\CR \subset \N$) $\CR =\N$  such that $\theta^*_i=\theta^*_j$ for any pair $i,j \in\CR$. In order to study the stability of the (partial) remote synchronization manifold, it suffices to study the stability of $\theta^*$. 
 In the next two sections, we investigate the local stability of remote synchronization manifolds.
  We start with the assumption that there is no phase shift in Section \ref{section:3}. The local stability of the remote and cluster synchronization manifolds is studied. In Section \ref{section:4}, we consider there is a phase shift $\alpha$ and investigate the influence of this phase shift on the stability of the remote synchronization manifold.

%%%%%%%%%%%%%%%%%%%%%%%%%%%%%%%%%%%%%%%%%%%%%%%%%%%%%%%%%%%%%%%%%%%%%%%%%%%%%%%%
\section{Remote Synchronization without Phase Shift}\label{section:3}
In this section, we consider the case when there is no phase shift, i.e., $\alpha=0$. We investigate how partial and complete remote synchronization in star networks are formed. We show the important roles that the symmetric outgoing couplings quantified by $A_i$ play in enabling synchronization among oscillators that are not directly connected.

To proceed, define $x_i=\theta_{0}-\theta_{i}$ for $i=1,2\dots,n$. Then the time derivative of $x_i$ is given by
\begin{align}
\dot x_i=\omega_{0}-\omega-\sum_{j=1}^{n}K_i \sin(\theta_0 -\theta_j)
- A_i \sin(\theta_0-\theta_i). \label{model:inc:EqNatFr} 
\end{align}
Let $x=[x_1,x_2,\dots,x_n]^\top\in \mathbb S^n, \bm{\omega}=(\omega_0-\omega)\mathbf{1}_n$ with $\mathbf{1}_n=[1,\dots,1]^\top$, and $\mat{sin}x=[\sin x_1,\dots,\sin x_n]^\top$, then the dynamics \eqref{model:inc:EqNatFr}  can be represented in a compact form as follows
\begin{align}
\dot x={\bm{\omega}} - T\mat{sin} x:=f(x), \label{model:inc:EqNatFr:comp}
\end{align}
where $f(x)=[f_1(x),\dots,f_n(x)]^\top$ and
\begin{align}
T= \left[ {\begin{array}{*{20}{c}}
	{A_1 + K_1}&K_2&\cdots&K_n\\
	K_1&{A_2 + K_2}& \cdots &K_n\\
	\vdots & \vdots & \ddots & \vdots \\
	K_1&K_2& \cdots &{A_n + K_n}
	\end{array}} \right]. \label{expr:T1}
\end{align}
Let $x^*$ be an equilibrium of \eqref{model:inc:EqNatFr:comp}, if it exists, i.e., $f(x^*)=0$. From the definition of $x$, we observe that $x^*$ corresponds to a (partial) remote synchronization manifold if there exists a set $\CR=\N$ ($\CR\subset \N$)  such that for any $i,j\in \CR$, $x^*_i=x^*_j$. In what follows, we show under what conditions on the coupling strengths the equilibrium $x^*$ exists and all (a part of) of its elements are identical, which gives rise to the corresponding (partial) remote synchronization of the model \eqref{Md:overall}. Towards this end, let us first make an assumption.

\begin{assumption}\label{assum:coup}
	We assume that the coupling strengths satisfy the following inequality 
	\begin{align}
	&A_i \ge (n-1) K_i,&\forall i\in\N, \label{ineq:ToDiaMatr}
	\end{align} 
	and the corresponding matrix $T$, given by \eqref{expr:T1}, satisfies
	\begin{align}
	\|T^{-1} \bm{\omega}\|_\infty<1. \label{Condi:EquNa}
	\end{align}
\end{assumption}

Assumption \ref{assum:coup} suggests that the strengths of outgoing couplings are much greater than that of incoming ones. By observing that for any $i$ it holds that
\begin{align*}
A_i+K_i-(n-1)K_i \ge& \\
 (n-1) &K_i+K_i-(n-1)K_i= K_i>0,
\end{align*}
we know that the matrix $T$ is \textit{column diagonally dominant}. By Gershgorin circle theorem \cite[Sec. 6.2]{horn1990matrix}, one knows all the eigenvalues of $T^\top$ have positive real parts, which also means that all the eigenvalues of $T$ lie on the right half plane. Thus $T$ is invertible.  We are now at a position to present our main result in this section.

\begin{theorem}\label{Theo:coupling}
	under Assumption \ref{assum:coup}, there exists a unique locally asymptotically stable equilibrium $x^*$ satisfying $|x^*_i|\in[0,\pi/2)$ for all $i\in\N$ for the dynamics \eqref{model:inc:EqNatFr}, which is
	\begin{align}
	x^*= \mat{arcsin}(T^{-1}{\bm{\omega}}). \label{equil:natEq}
	\end{align}
	 {In addition, if there is a pair of distinct $i,j\in\N$ such that  $A_i=A_j$, then $x^*_i=x^*_j$.} This $x^*$ corresponds to a partial remote synchronization manifold, denoted by $\Rot(\theta^*)$, for the dynamics \eqref{Md:overall}, which implies oscillators $i$ and $j$ are remotely synchronized.
	\end{theorem}

	\begin{proof} {Due to the page limit, we only provide a sketch
		of proof. From the hypothesis \eqref{Condi:EquNa}, one knows that there exists a unique solution in $[0,\pi/2)$, i.e., $x^*= \mat{arcsin} (T^{-1} {\bm{\omega}})$.
	
	We show the stability of this equilibrium by linearization. Towards this end, we evaluate the Jacobian matrix
	\begin{align}
	J\left( x^* \right) = { - {{\left. {\frac{{\partial f}}{{\partial x}}} \right|}_{x = {x^*}}}} 
	=-T\diag\left(\cos x^*_1,\dots,\cos x^*_n\right). \label{Joca:expr}
	\end{align}
	By using the facts that $T$ is column diagonally dominant and $\cos x^*_i>0$ for all $i\in\N$, it is not hard to see that all the eigenvalues of $J(x^*)$ have negative real parts, which means the equilibrium $x^*_i$ is locally asymptotically stable.
	
	Finally, we prove $x^*_i=x^*_j$ if $A_i=A_j$ by showing that $\sin x^*_i=\sin x^*_j$ since $|x^*_i|<\pi/2$ for all $i$.}
	\end{proof}

	\begin{remark}	
	{Theorem \ref{Theo:coupling} shows that oscillators without direct connections are able to get phase synchronized just because the directed connections from the central mediator towards them are identical. Authors in \cite{nicosia2013remote} show that symmetries in undirected networks are important for remote synchronization. In contrast, we take directions of the couplings into account, and show that only the outgoing couplings, $A_i$, matter. This idea is analogous to equitable partition in graphs investigated in \cite{schaub2016graph} and \cite{tiberi2017synchronization} if the central mediator is regarded as a cluster and the peripheral ones as another. Different from these two works, we studied the stability of remote synchronization.
	}
	\end{remark}
	
	What is worth mentioning, by carefully manipulating the symmetry of the couplings originated from the central node $0$, not only synchronization among distant oscillators can be facilitated, but also unnecessary synchronization can be easily prevented. Moreover, interesting patterns of remote synchronization, such as cluster synchronization, can occur. The following corollary provides some sufficient conditions for the existence and stability of remote and cluster synchronization manifold, which follows from Theorem \ref{Theo:coupling} straightforwardly.
	
	\begin{corollary}\label{Coro:1}
		Under Assumption \ref{assum:coup}, there is a locally asymptotically stable remote synchronization manifold for the dynamics \eqref{Md:overall}, i.e., in which the solution $\theta(t)$ is completely remotely synchronized, if $A_i=A_j$ for every pair $i,j\in\N$; there is a locally asymptotically stable partial remote synchronization manifold for the dynamics \eqref{Md:overall}, in which the solution $\theta(t)$ is  remotely clustered with respect to $\mathcal C$, if there is a partition of $\N$, denote by $\mathcal C=\{\mathcal C_1,\dots,\mathcal C_m\}$, satisfying $|\mathcal C_p|\ge 2$ and $ \cup_{p=1}^{m}\mathcal C_p=\N$ such that: for any given $\mathcal C_p$ and every pair $i,j\in \mathcal C_p$ it holds that $A_i=A_j$; on the other hand, for any given $i\in \mathcal C_p, j \in \mathcal C_q$ where $p\neq q$, $A_i\neq A_j$.
	\end{corollary}
		
	In next section, we consider the case where there is a phase shift (or phase lag) term $\alpha$. The model with the presence of a phase shift is known as the Kuramoto-Sakaguchi model \cite{sakaguchi1986soluble}. 

%%%%%%%%%%%%%%%%%%%%%%%%%%%%%%%%%%%%%%%%%%%%%%%%%%%%%%%%%%%%%%%%%%%%%%%%%%%%%%%%
\section{Remote Synchronization with Phase Shift}\label{section:4}
In this section, we consider that there is a phase shift $\alpha \in (0,\pi/2)$. By introducing a phase shift term, we allow the oscillators to get synchronized at a frequency that differs from the average of their natural frequencies \cite{montbrio2004synchronization}. This phenomenon have always been observed in many biological systems such as heart cells and mammalian intestine  \cite{winfree2001geometry}. Moreover, in neural networks the phase shift $\alpha$ is often used to model delays concerning synaptic connections \cite{hoppensteadt2012weakly}.
To study the remote synchronization of our interest, we simplified the problem by assuming that $A_i=A$ and $K_i=A/n$ for all $i$. Note that this simplification ensures that the direction of the network is preserved and the condition \eqref{ineq:ToDiaMatr} is satisfied, which guarantees the property that the outgoing couplings are much stronger than the incoming ones. 
Consequently, the dynamics \eqref{Md:overall} become
\begin{align}
&\dot \theta_0 = \omega_0+ \frac{A}{n} \sum_{i=1}^{n}\sin (\theta_i-\theta_0-\alpha); \nonumber\\
&\dot \theta_i=\omega+A\sin(\theta_0-\theta_i-\alpha),i=1,2,\dots,n. \label{Md:phaShi}
\end{align}
 Conditions on the coupling strength $A$ are subsequently obtained to ensure that the dynamics \eqref{Md:phaShi} admit a locally asymptotically stable remote synchronization manifold. 
We investigate how these conditions depend on the phase shift $\alpha$. As frequency synchronization is the footstone for the analysis that follows, let us provide the necessary condition for the existence of a frequency synchronized solution to \eqref{Md:phaShi} and see how it depends on the phase shift $\alpha$.
\begin{proposition} \label{necessary}
	There is a frequency synchronized solution to the dynamics \eqref{Md:phaShi} only if \[A\ge \frac{1}{2\cos \alpha} |\omega_0-\omega|.\]
\end{proposition}

{Letting $\dot\theta_0-\dot\theta_i=0$ for all $i$, the proof of Proposition \ref{necessary} follows immediately.} We observe that when $\alpha=0$, this necessary condition reduces to $A\ge  |\omega_0-\omega|/2$. Obviously, the existence of the phase shift raises the requirement for the coupling strength $A$. Next, we show the sufficient conditions on $A$ such that there is a locally asymptotically stable remote synchronization manifold for \eqref{Md:phaShi}. Towards this end, let $y_i=(\theta_{0}-\theta_{i})/2, y_i\in\mathbb{S}^1$ for $i=1,2\dots,n$. The time derivative of $y_i$ is 
\begin{align}
\dot y_i=&\frac{1}{2}(\omega_0-\omega)+\frac{A}{2n} \sum_{j=1}^{n}\sin (\theta_j-\theta_0-\alpha)	\nonumber\\
&\;\;\;\;\;\;\;\;\;\;\;\;\;\;\;\;\;\;\;\;\;\;\;\;\;\;\;\;\;\;\;\;\;\;\;\;\;-\frac{1}{2}A\sin(\theta_0-\theta_i-\alpha)\nonumber\\
=&\frac{1}{2}(\omega_0-\omega)-\frac{A}{2n} \sum_{j=1}^{n}\sin (2y_j+\alpha)\nonumber\\
&\;\;\;\;\;-\frac{1}{2}A\sin(2y_i-\alpha)
:=g_i(y),i=1,2,\dots,n. \label{der:yi:ori}
\end{align}
where $y=[y_1,\dots,y_n]^\top$ and $g(y)=[g_1(y),\dots,g_n(y)]^\top$. 
 Let us provide the main result in this section. 
\begin{theorem}\label{Theo:withAL}
	There is a unique locally asymptotically stable  equilibrium $y^*$ for the dynamics \eqref{der:yi:ori} satisfying $|y^*_i| < \pi/4$ for all $i$, which is 
	\begin{align}
	y^*=\frac{1}{2}\arcsin \left(\frac{\omega_0-\omega}{2A\cos \alpha}\right)\mathbf{1}_n, \label{equil:ComRemo}
	\end{align}
	if the following conditions are satisfied, respectively:
	\begin{enumerate}
		\item[i)] when $\omega_0>\omega$, the coupling strength $A$ satisfies  
			\begin{align}
		A>\frac{\omega_0-\omega}{2\cos\alpha}; \label{Posi:condi}
		\end{align} 
		\item[ii)] when $\omega_0<\omega$, the coupling strength $A$ satisfies  
		\begin{align}
		A>\frac{\omega-\omega_0}{2\cos^2\alpha}.\label{Nega:condi}
		\end{align} 
	\end{enumerate}
This locally asymptotically stable  equilibrium $y^*$ for the dynamics \eqref{der:yi:ori} corresponds to the locally asymptotically stable remote synchronization manifold for \eqref{Md:phaShi}.
\end{theorem}
	\begin{proof}
		{The proof is similar to that of Theorem \ref{Theo:coupling}. The idea is to evaluate the Jacobian matrix $J(y)$ at the equilibrium $y=y^*$. We prove that all the eigenvalues of $J(y)$ have negative real parts by showing that $J(y)$ has negative diagonal elements and is diagonally dominant if hypotheses \eqref{Posi:condi} and \eqref{Nega:condi} are satisfied for the cases i) and ii) respectively. This implies the stability of the equilibrium $y^*$.}	
	\end{proof}
	\begin{remark}
	Theorem \ref{Theo:withAL} provides some sufficient conditions for the existence and local stability of the equilibrium of dynamics \eqref{der:yi:ori}, or equivalently, for the existence and local stability of remote synchronization manifold of \eqref{Md:phaShi}. With the presence of the phase shift $\alpha$, the requirement of coupling strengths is increased. In fact, the larger the phase shift is, the stronger the coupling is required, which can be observed from \eqref{Posi:condi} and \eqref{Nega:condi}. Interestingly, comparing \eqref{Nega:condi} with \eqref{Posi:condi} we observe that the phase shift has a different impact on the coupling strength in the two cases when $\omega_0>\omega$ and $\omega_0<\omega$. The latter case is more vulnerable to phase shift.
	\end{remark}
	
	%%%%%%%%%%%%%%%%%%%%%%%%%%%%%%%%%%%%%%%%%%%%%%%%%%%%%%%%%%%%%%%%%%%%%%%%%%%%%%%%
	\section{Numerical  Examples} \label{simulation}
	To validate the results we obtained in Section \ref{section:3} and Section \ref{section:4}, we perform some numerical studies in this section. We consider $7$ oscillators coupled by a directed star network illustrated in Fig. \ref{Fig:star}. To measure the levels of synchronization we introduce the two useful functions as follows,
	\begin{align*}
	h_1(\theta(t))=\max_{i,j\in\N}|\theta_i(t)-\theta_j(t)|,\\
	h_2(\theta(t))=\max_{i\in\N}|\theta_0(t)-\theta_i(t)|,
	\end{align*} 
	If $h_2=0$, the phase difference between any peripheral oscillator and the central one is zero, which implies complete synchronization in the whole network. In particular, if $h_1=0,h_2\neq 0$, all the phases of peripheral oscillators are identical remaining central one different, which yields remote synchronization.
	
		\begin{figure}[!t]
		\centering
		{
			\begin{tikzpicture} [->,>=stealth',shorten >=1pt,auto,node distance=1.2cm,
			main node/.style={circle,fill=blue!5,draw,minimum size=0.4cm,inner sep=0pt]},
			red node/.style={circle,fill=red!30,draw,minimum size=0.4cm,inner sep=0pt]}]
			
			%		[>=stealth',shorten >=1pt,node distance=3cm,on grid,initial/.style=]
			\node[red node]          (0) at (0,0)                      {$0$};

			\node[main node]          (1) at (0,1.5)         {$1$};		
			\node[main node]          (2) at (1.35,0.8)        {$2$};	
			\node[main node]          (3) at (1.35,-0.8)          {$3$};
			\node[main node]          (4) at (0,-1.5)          {$4$};	
			\node[main node]          (5) at (-1.35,-0.8)         {$5$};
			\node[main node]          (6) at (-1.35,0.8)         {$6$};
			
			%		\draw	(0,-0.25) node[anchor=west] {0}; 

			\tikzset{direct/.style={->,line width=1pt}}
			\tikzset{red_direct/.style={->,line width=1pt,red}}
			\path (0)     edge[direct]     node   {} (1)
			(1)     edge[direct]     node   {} (0)		 
			(0)     edge[direct]     node   {} (2)  
			(2)     edge[direct]     node   {} (0)  
			(0)     edge[direct]     node   {a} (3)  
			(3)     edge[direct]     node   {a} (0) 
			(0)     edge[direct]     node   {} (5)
			(5)     edge[direct]     node   {} (0)
			(0)     edge[direct]     node   {} (6)
			(6)     edge[direct]     node   {} (0)
			(0)     edge[direct]     node   {} (4)
			(4)     edge[direct]     node   {} (0);		
			\end{tikzpicture}}		
		\caption{ The star network considered: central node $0$ and peripheral ones $\{1,2,3,4,5,6\}$.}
		\label{Fig:star}	\end{figure}
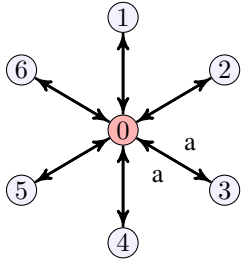

\begin{figure}[!t]		
	\centering{
		\includegraphics[width=0.25\textwidth]{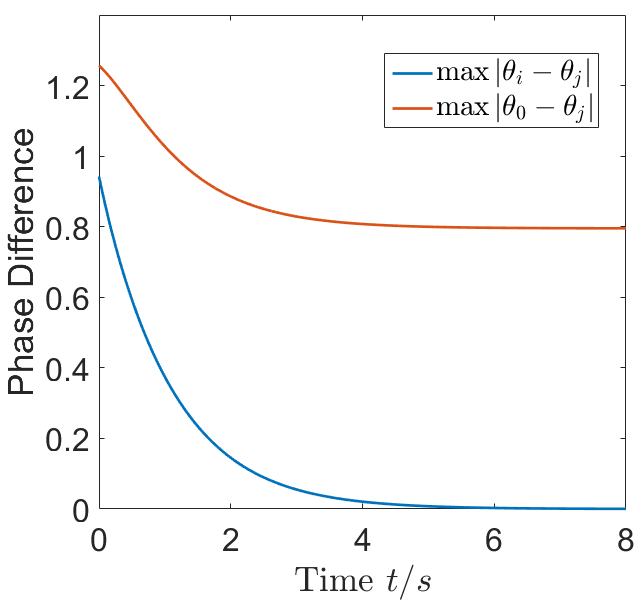}}
	\label{Th1_Comple}
	\caption{Trajectories of the maximum absolute values of the phase differences when $\alpha=0$: blue represents $h_1=\max_{i,j\in\N}|\theta_i-\theta_j|$ and red represents $h_2=\max_{i\in\N}|\theta_0-\theta_i|$.}
	\label{Trj:CompRemot}
\end{figure}
	
	We first testify the results obtain in Theorem \ref{Theo:coupling}.  To distinguish the frequencies, let the frequency of each peripheral oscillator be $\omega=0.8\pi \; rad/s$, and the natural frequency of the central one be $\omega_0=1.5\pi \; rad/s$. In order to make complete remote synchronization occur, we let $A_i=1.4$ for all $i=1,2,\dots,6$, and let $K_1=0.15, K_2=0.12,K_3=0.2, K_4=0.18, K_5=0.2,  K_6=0.14$. Then the matrix $T$ becomes		
	\begin{align*}
	T=	\left[ {\begin{array}{*{20}{c}}
		{1.55}&{0.12}&{0.2}&{0.18}&{0.2}&{0.14}\\
		{0.15}&{1.52}&{0.2}&{0.18}&{0.2}&{0.14}\\
		{0.15}&{0.12}&{1.6}&{0.18}&{0.2}&{0.14}\\
		{0.15}&{0.12}&{0.2}&{1.58}&{0.2}&{0.14}\\
		{0.15}&{0.12}&{0.2}&{0.18}&{1.6}&{0.14}\\
		{0.15}&{0.12}&{0.2}&{0.18}&{0.2}&{1.54}
		\end{array}} \right].
	\end{align*}	
	It can be verified that $T$ is diagonal dominated and $|T^{-1}\bm{\omega}|=0.9201<1$, i.e. conditions in Assumption \ref{assum:coup} are satisfied. Let the initial phases be $\theta(0)=[1.3\pi,1.2\pi,1.15\pi,0.9\pi,1.2\pi,1.0\pi,1.11\pi]^\top$, and then the trajectories of $h_1(\theta(t))$ and $h_2(\theta(t))$ are presented in Fig. \ref{Trj:CompRemot}. It can be observed that $h_1(\theta(t))$ converges to zero, while $h_2(\theta(t))$ converges to a constant, suggesting that the peripheral oscillators which are not directly connected achieve phase synchronization, but the ones that have direct connections (the central one with each peripheral one) do not.  Next, we show that cluster synchronization is formed if the conditions in Corollary \ref{Coro:1} are satisfied. Let the outgoing coupling strengths be $A_1=A_4=2.1, A_2=A_5=2.8, A_3=A_6=4.2$, and let the incoming coupling strength be the same as considered above. One can also check Assumption \ref{assum:coup} is satisfied since $|T^{-1}\bm{\omega}|=0.7743<1$. Let $\theta(0)=[1.3 \pi, 0.2\pi, 0.6\pi, 1\pi, 1.4 \pi, 1.8\pi, 2\pi]^\top$, and the phases of the oscillators are plotted on the  unit circle $\mathbb{S}^1$ at a sequence of time instants (see Fig. \ref{snapshot}). One can observe that the peripheral oscillators with the same outgoing strength $A_i$ get phase synchronized, forming three clusters (in each of which phases are different from the central one's). This suggests that the symmetry of the outgoing couplings of the peripheral oscillators plays an essential role in facilitating remote synchronization. 
	
		\begin{figure}[t!] 
		\centering
		\subfigure[$t=0$]{
			\includegraphics[width=0.12\textwidth]{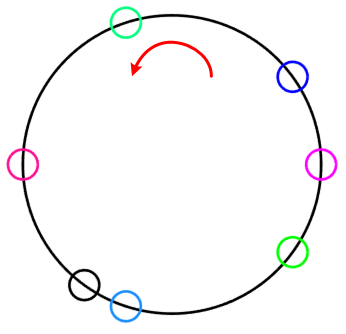}}
		\subfigure[$t=0.5$]{
			\includegraphics[width=0.12\textwidth]{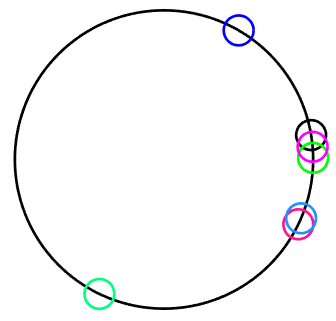}}
		\subfigure[$t=1$]{
			\includegraphics[width=0.12\textwidth]{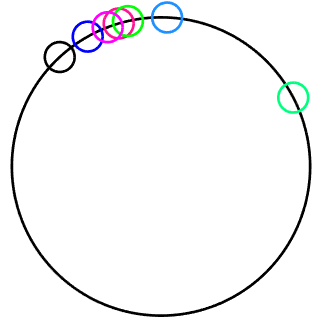}}
		\subfigure[$t=2$]{
			\includegraphics[width=0.12\textwidth]{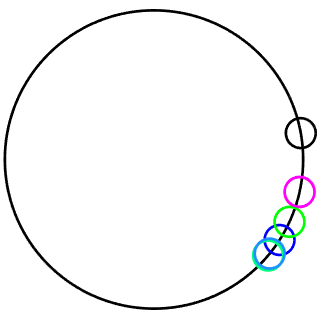}}
		\subfigure[$t=4$]{
			\includegraphics[width=0.12\textwidth]{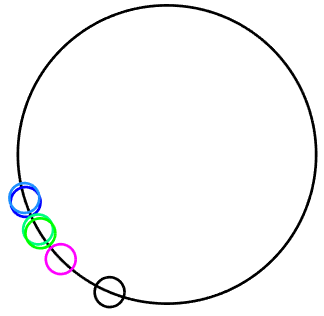}}
		\subfigure[$t=8$]{
			\includegraphics[width=0.12\textwidth]{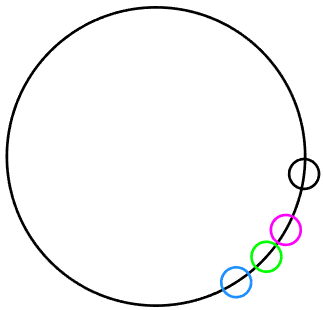}}
		
		\caption{The phases on $\mathbb{S}^1$ at six time instants when $\alpha=0$: black represents the central oscillator $0$; blue represents oscillators $1$ and $4$; green represents $2$ and $5$; red represents $3$ and $6$.}
		\label{snapshot}	
	\end{figure}

\begin{figure}[!t]		
	\centering
	\subfigure[$\omega_0<\omega,A=1.6\pi$]{		
		\includegraphics[width=0.20\textwidth]{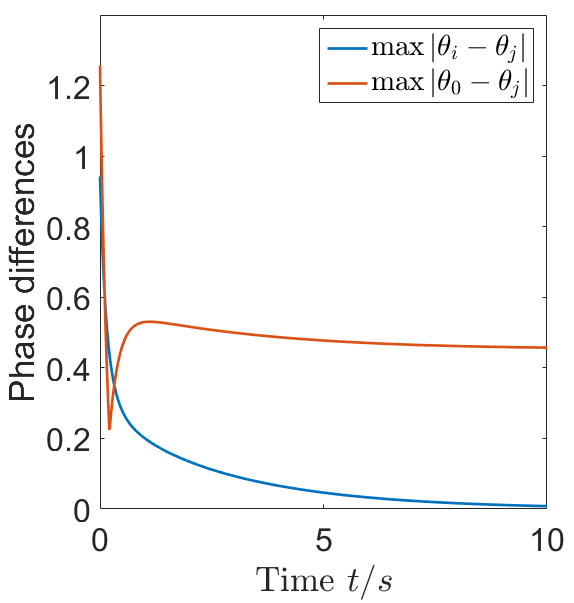}
		\label{phase_D-mi_syn}}	
	\subfigure[$\omega_0<\omega,A=\pi$]{		
		\includegraphics[width=0.20\textwidth]{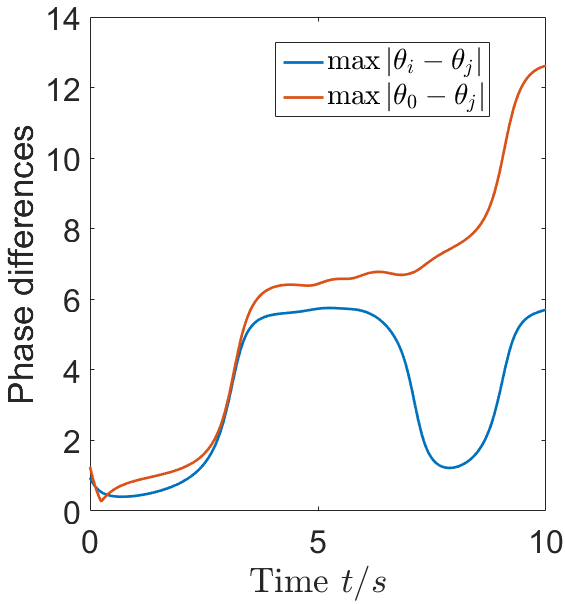}
		\label{phase_D-mi}}
	\subfigure[$\omega_0>\omega,A=0.8\pi$]{		
		\includegraphics[width=0.20\textwidth]{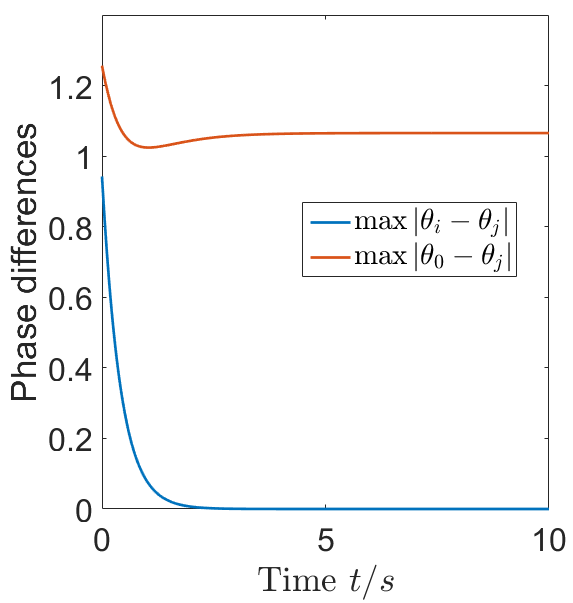}
		\label{phase_D-posi_syn}}	
	\subfigure[$\omega_0>\omega,A=0.4\pi$]{		
		\includegraphics[width=0.20\textwidth]{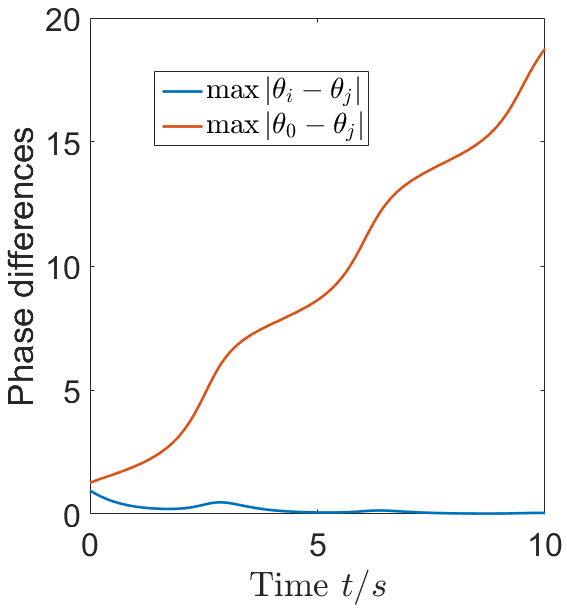}
		\label{phase_D-posi}}		
	\caption{Trajectories of the maximum absolute values of the phase differences when $\alpha=\pi/3$: blue represents $h_1=\max_{i,j\in\N}|\theta_i-\theta_j|$ and red represents $h_2=\max_{i\in\N}|\theta_0-\theta_i|$.}
	\label{Trj:CompThe22}
\end{figure}

	Next, we validate the results in Section \ref{section:4}, where there is a phase shift $\alpha$. Without loss of generality, let $\alpha=\pi/3$.
	First, we consider the case when $\omega_0<\omega$. Let the frequency of each peripheral oscillator be $\omega=0.8\pi$, and the natural frequency of the central one be $\omega_0=0.1\pi$.
	From and the condition \eqref{Nega:condi}, we calculate the threshold of the coupling strength $A$, which is $(\omega-\omega_0)/(2\cos^2 \alpha)=1.4\pi$. Let $A=1.5\pi>1.4\pi$, and we plot the absolution value of phase differences $h_1(\theta(t))$ and $h_2(\theta(t))$ in Fig \ref{phase_D-mi_syn}, from which we observe that remote synchronization is achieved. On the contrary, if we let $A=\pi$, it can be seen from Fig. \ref{phase_D-mi} that remote synchronization does not occur. Finally, we consider the case $\omega_0>\omega$ by  letting $\omega_0=1.5\pi,\omega=0.8\pi$. The threshold given in \eqref{Posi:condi} becomes $(\omega-\omega_0)/(2\cos \alpha)=0.7\pi$. The trajectories of $h_1(t)$ and $h_2(t)$ when $A=0.8 \pi$ and $A=0.4 \pi$ are presented   in Fig. \ref{phase_D-posi_syn} and \ref{phase_D-posi}, respectively. Shown is Fig. \ref{phase_D-posi_syn}, remote synchronization is achieved. Surprisingly, one can observe from Fig. \ref{phase_D-posi} that the phase differences among peripheral oscillators approach zero, although the phase differences between the peripheral and the central oscillators are increasing.  This implies remote synchronization can also take place without requiring that all the frequencies get synchronized. We are currently working on the construction of a mathematical proof for this observed phenomenon.

	\section{Conclusion}\label{section:5}

	Motivated by synchronization observed in distant cortical regions in human brain, especially neuronal synchrony of unconnected areas through relaying, we have studied remote synchronization of Kuramoto oscillators  coupled by a star networks. We have shown that the symmetry of outgoing connections from the central oscillator plays a critical role in facilitating phase synchronization between peripheral oscillators. By carefully adjusting the strengths of these couplings, interesting patterns of stable remote synchronization, such as cluster synchronization, can be achieved. We have also studied the case when there is a phase shift. Sufficient conditions have been obtained to ensure the stability of remote synchronization. Simulations have been performed to validate our results. We are highly interested in generalizing our results on remote synchronization to more complex networks in the future. 

%%%%%%%%%%%%%%%%%%%%%%%%%%%%%%%%%%%%%%%%%%%%%%%%%%%%%%%%%%%%%%%%%%%%%%%%%%%%%%%%
\bibliographystyle{IEEEtran}
\bibliography{IEEEabrv,papers}

\end{document}